\documentclass[11pt]{article} \usepackage{amssymb}
\usepackage{amsfonts} \usepackage{amsmath} \usepackage{amsthm} \usepackage{bm}
\usepackage{latexsym} \usepackage{epsfig}
\usepackage{url}
\usepackage{hyperref}
\newtheorem*{theorem*}{Theorem}
\newtheorem{theorem}{Theorem}[section]

\newtheorem*{proposition*}{Proposition}

\newcommand{\ignore}[1]{}

\newcommand{\enote}[1]{} \newcommand{\knote}[1]{}
\newcommand{\rnote}[1]{}

\renewcommand{\P}[1]{{\mathbb{P}}\left[{#1}\right]}

\newcommand{\E}[1]{{\mathbb{E}}\left[{#1}\right]}
\newcommand{\ind}[1]{{\bf 1}_{#1}}

\newcommand{\CondE}[2]{{\mathbb{E}}\left[{#1}\middle\vert{#2}\right]}

\newcommand{\eps}{\epsilon}

 \newcommand{\R}{\mathbb R}

\renewcommand{\phi}{\varphi}

\begin{document}
\title{Bundling Customers: How to Exploit Trust Among Customers to Maximize Seller Profit}
\author{
  Elchanan Mossel\footnote{Weizmann Institute and U.C. Berkeley. E-mail:
    mossel@stat.berkeley.edu. Supported by a Sloan fellowship in
    Mathematics, by BSF grant 2004105, by NSF Career Award (DMS 054829)
    by ONR award N00014-07-1-0506 and by ISF grant 1300/08}~ and
  Omer Tamuz\footnote{Weizmann Institute. Supported by ISF
    grant 1300/08. Omer Tamuz is a recipient of the Google Europe
    Fellowship in Social Computing, and this research is supported in
    part by this Google Fellowship.}
 }

\date{\today}
\maketitle
\begin{abstract}
  We consider an auction of identical digital goods to customers whose
  valuations are drawn independently from known distributions.
  Myerson's classic result identifies the truthful mechanism that
  maximizes the seller's expected profit.

  Under the assumption that in small groups customers can learn each
  others' valuations, we show how Myerson's result can be improved to
  yield a higher payoff to the seller using a mechanism that offers
  groups of customers to buy bundles of items.
\end{abstract}

\section{Introduction}

\subsection{Second version note}
After posting the first version of this paper we learned that much of
its mathematical content already appears in the literature, for
example in an article of
Armstrong~\cite{armstrong1996multiproduct}, although in a slightly
different context of bundling products, rather than customers.

\subsection{Bundling items}
Bundling is the practice of joining together a number of products into
a ``bundle'', so that customers may not buy each product separately,
but must choose to either buy the entire bundle or have non of the
included items. Alternatively, customers may be allowed to purchase a
single item, but at a higher cost; that is, the price of the bundle is
set to below the sum of the prices of the individual items that
comprise it. Examples range from McDonald's happy meals to enormous
defense contracts~\cite{adams1972military} (see also the recent
attention to bundling of scientific journal
subscriptions~\cite{gowersblogElsevier}). Bundling has also received
much attention from theorists (cf.~\cite{adams1976commodity,
  mcafee1989multiproduct, holzman2004bundling} and many more).

However, consider a population of consumers who are potential
customers for some mass produced product (i.e., the number of
available items is unlimited). Assume also that customers generally
have no need for more than one item. For example, the product might be
an upgrade to an operating system, a cellphone data package or removal of 
tax offenses record. This class of products is sometimes referred to as
{\em digital goods}.

Since each customer has no need for more than one item, bundling items
does not seem to offer an advantage to the seller. Indeed,
Myerson~\cite{myerson1981optimal} shows that in a Bayesian setting the
best strategy available to the seller is to offer a fixed per-item
price\footnote{Different fixed prices may be offered to different
  customers in a practice called {\em price differentiation}.}. Since
customer valuations for a product may differ wildly, fixing a price
often means forfeiting the customers who are willing to pay less,
while undercharging the customers who are willing to pay more.

\subsection{Bundling customers}
We consider a different kind of bundling, which, although also
widespread, seems (to our knowledge) to have been largely overlooked
by theorists.  We propose that the seller may increase its profit
beyond Myerson's bound by {\em bundling customers}: here customers are
arbitrarily grouped into pairs and are offered to buy two items for a
price that is lower than the sum of the prices of the individual
items. The same can of course be done for larger groups of customers,
so that a group of $n$ customers are jointly offered to buy $n$ items
for a discount.

Our key assumption is what we call {\em group rationality}: namely
that a bundle of customers will accept the group offer {\em if there
  is a way for them to share the cost so that all of them
  benefit}. For example, consider two customers who are each
interested in buying a copy of a book whose (single item) price is set
to $10$ Gold Dinars. Let customer $X$ be willing to pay at most $20$
Dinars, and let customer $Y$ be willing to pay at most $5$ Dinars. Let
the cost of a single book be set to $10$ Dinars. Group rationality
implies that if $X$ and $Y$ were offered to jointly buy two books for
$11$ Dinars then they would accept and find a way to split the cost,
since both can benefit; customer $X$ can contribute $8$ Dinars and
customer $Y$ can contribute $3$ Dinars, and then $X$ has paid two 
Dinars less than she would have paid on her own, and $Y$ was able to
buy the book, which he wouldn't been able to do on his own. Note that
assuming that the cost of printing a book is small, then the seller is
also strictly better off.

Our {\em group rationality} assumption is novel in the context of
Myerson auctions, and is in fact what allows us to increase the
seller's profit past Myerson's bound on truthful auctions. We note
that indeed there is no truthful mechanism for two customers to agree
on a division of costs when a feasible one exists; this is nothing but
the well known ``splitting the dollar'' game. In the example above, if
customer $Y$ manages to convince $X$ that he is not willing to pay
more than $2$ Dinars, then $X$ might settle for paying $9$ herself,
which still leaves her better off than buying a single book for $10$
Dinars.

However, we argue that it is important to consider group rationality;
it is in fact a phenomenon that, in other contexts, has been widely
studied theoretically and experimentally, and falls under the general
titles of cooperation and altruism (cf.~\cite{axelrod1981evolution,
  nowak2006five, bowles2006group, nowak2005evolution}).

Specifically, families and tribes are often group rational (for
obvious evolutionary reasons, cf.~\cite{kurland1980kin}), as are other
groups of people who expect to have to rely on each other in the
future.  A further argument to support group rationality in our
setting is the observation that when the stakes (i.e., the savings)
are high, one could expect that in any small group people would be
sufficiently incentivized to find a way to compromise, trust and
share, even if there is a danger of being short-changed; in reality,
prisoners do sometimes choose to ``cooperate'' even when facing the
risk of ``defection'' by cellmates, and the tragedy of the commons can
be averted (cf.~\cite{fehr2000cooperation}).

\subsection{Results}
We consider a Bayesian setting with independent customer valuation
distributions and {\em group rational} customers. Our main result is that
under mild smoothness conditions of the customers' valuation
distributions, the seller can expect a strictly higher profit when
bundling customers into pairs, as compared to selling single items.

We also show that when valuations are uniformly bounded then, as the
size of the bundle increases, the seller's expected profit from the
customers approaches the sum of their expected valuations for the
product, which is an upper bound on the seller's profit. This bound is
achieved in single customer auctions only when the customer reveals
its valuation to the seller.

Approaching this limit by bundling ever larger groups of customers
would require ever more trust among them. Note that assuming group
rationality for larger groups is a stronger assumption than group
rationality for smaller groups.  Indeed, as the size of the group
grows, the believability of {\em group rationality} diminishes; all
else being equal, it seems harder to expect honesty and trust among a
hundred people than among a couple.

Our results can therefore be interpreted to show that the seller can
exploit trust {\em among customers}\footnote{Note that the customers
  are not required to trust the seller!} to increase its profit. And
in fact, the more trusting the customers are (i.e., the larger the
trusting group is), the higher the profit the seller can expect, up to
the maximal profit possible.

\section{Model}
Let $[n] = \{1,\ldots,n\}$ be the set of customers.  Each customer $i$
has a private valuation $V_i$, which is the maximum price that it
would be willing to pay for the product. These valuations are not
known to the seller, who however has some knowledge of what they might
be. We model the seller's uncertainty by assuming that each valuation
$V_i$ is picked independently\footnote{Despite some recent
  progress~\cite{papadimitriou2011optimal}, it seems that Myerson
  auctions are generally difficult to analyze when valuations are not
  independent.  We conjecture that our results hold also for the case
  of correlated valuations.} from some distribution with cumulative
distribution function (CDF) $F_i$. This model is a special case of
Myerson auctions~\cite{myerson1981optimal}.

We make a number of mild smoothness conditions on the distributions of
valuations: We assume that $F_i$ is non-atomic and differentiable with
bounded density (PDF) $f_i$. We assume all valuations are in $[0,M]$
for some $M \in \R$, so that $f_i$ is zero outside this interval for
all $i$. We further assume that for some $\delta > 0$ it holds for all
$i$ that $\delta < f_i < 1/\delta$ in the interval
$[0,M]$\footnote{These assumptions can be significantly relaxed at the
  price of a significantly more technical and difficult to read
  paper.}.

Let $s$ be an auction mechanism or sales strategy. We assume that it
can result in each of the customers either receiving or not receiving
an item, and parting with some sum of money. In the context of $s$, we
denote by $R_i^s$ the event that customer $i$ receives an item. We
denote by $P_i^s$ the price, or the amount of money customer $i$ paid
the seller for the item. We denote customer $i$'s utility by $C_i^s$,
where
\begin{align}
  \label{eq:customer-utility}
  C_i^s = \ind{R_i^s}(V_i-P_i^s).
\end{align}
and denote the customer's expected utility by $c_i^s=\E{C_i^s}$. 

Let $U_i^s$ denote the seller's utility from selling to customer
$i$. We assume that the cost of an item is zero, and so define
\begin{align}
  \label{eq:seller-utility}
  U_i^s = \ind{R_i^s}P_i^s.
\end{align}
We denote the seller's expected utility by $u_i^s=\E{U_i^s}$. We
denote the seller's total expected utility by $u^s=\sum_{i=1}^nu_i^s$.

We assume throughout that given a seller's strategy, the customer will
pick a strategy that will maximize its expected utility. Given that, a
seller will pick a strategy that will maximize its own total expected
utility. We largely ignore the possibility of ties (i.e., two
strategies that result in the same expected utility, for either the
customer or the seller), since, as we assume the distribution of the
valuations is non-atomic, it will be the case for the strategies that
we consider that ties will occur with probability zero.

\subsection{Sales strategies}
\subsubsection{Single customer one time offer}
\label{sec:single-model}

We assume that the seller wishes to maximize the sum of the expected
revenues it extracts from the customers. A possible strategy would be
to give each customer $i$ a one time offer to buy the product at price
$p_i$. Myerson~\cite{myerson1981optimal} shows that this sales
strategy, of all the truthful strategies, maximizes the profit of the
seller, for the appropriate choice of $p_i$.

The customer's utility in this case is
$C_i=(V_i-p_i)\ind{R_i}$. Therefore, assuming the customer wishes to
maximize its utility, it would buy iff $V_i \geq p_i$ (or equivalently
$V_i > p_i$, since $\P{V_i=p_i}=0$). Hence $R_i=\ind{V_i \geq
  p_i}$, the gain by the seller is $P_i=p_i\ind{R_i}$, and the
seller's expected utility is
\begin{align}
  \label{eq:u}
  u_i(p_i) = \E{P_i} = p_i \cdot \P{V_i \geq p_i} + 0 \cdot \P{V_i < p_i}=
  p_i(1-F_i(p_i)),
\end{align}
with
\begin{align}
  \label{eq:dudp}
  u_i'(p_i) = \frac{du_i(p_i)}{dp_i} = 1-F_i(p_i)-p_if_i(p_i).
\end{align}

If we assume that $F_i$ is non-atomic, differentiable and only supported
on $[0,M]$, then $u_i(p_i)$ is continuous and differentiable and must
have a maximum in $[0,M]$. By solving $u_i'(p_i) = 0$ we can show that
any $p_i$ which maximizes $u_i$ satisfies
\begin{align}
  \label{eq:p_i-opt}
  p_i = \frac{1-F_i(p_i)}{f_i(p_i)}.
\end{align}
Furthermore, under these assumptions $u_i(0)=u_i(M)=0$, whereas
clearly $u_i$ is positive for some $0 < p_i < M$. Hence this maximum
does not occur at $0$ or $M$.

\subsubsection{Bundling customers}
We next consider the strategy of bundling the customers. Let $[n] =
\{1,\ldots,n\}$ be a set of customers. The bundling strategy here is
parametrized by a vector of single item prices
$\bar{a}=(a_1,\ldots,a_n)$ and the bundle price $b$.

The $n$ customers are given the option to buy a bundle of $n$ items
(i.e., each gets an item) for the total price of $b$. Additionally,
each customer $i$ may buy a single item for the price of $a_i$.

We assume {\em group rationality}, so that the customers choose to buy
the bundle {\em if the cost can be shared in a way that is profitable
  for all}. That is, the customers buy the bundle if {\em there exist}
$(P_1,\ldots,P_n)$ such that the following holds:
\begin{enumerate}
\item $\sum_iP_i=b$.
\item $P_i \leq V_i$ for all $i \in [n]$. That is, each customer's utility
  for buying the bundle is positive, or better than the utility for
  not buying.
\item $P_i \leq a_i$ for all $i \in [n]$. That is, each customer's
  utility for buying the bundle is better than the utility for buying
  individually.
\end{enumerate}
Hence, we assume that if the cost of the bundle can be shared in a way
that, for each customer, improves the utility over the other
alternatives, then the customers will find a way to share the cost and
will choose to buy the bundle. When this is not the case then each
customer $i$, independently, decides to either buy or to buy,
depending on whether $a_i \leq V_i$, as in the single customer
case. Formally, $R_i=1$ iff the condition above holds or $a_i \leq
V_i$.  Note that ``$\leq$'' can be replaced by ``$<$'' throughout,
since ties occur with probability zero.

Note that when the conditions above apply - i.e., accepting the bundle
is group rational for some prices $\{P_i\}$ - then accepting the offer
and paying $P_i$ is a Nash Equilibrium: it is better for customer $i$
to accept the offer for $P_i$ rather than shop alone, since then it
would have to pay more.

\section{Results}
\subsection{Smoothness and boundedness conditions}
\label{sec:smoothness}
We make the following assumptions on the distribution of customer
valuations $V_i$. Recall that we denote by $F_i$ and $f_i$ the CDF and
PDF of the distribution of $V_i$.
\begin{enumerate}
\item Customers valuations are independent and non-atomic.
\item There exists $M > 0$ such that, for all $i$, $V_i$ is in
  $[0,M]$.
\item The distribution of $V_i$ has a density (PDF) $f_i$.
\item There exists $0 < \delta < 1$ such that, for all $i$, $\delta < f_i(p)
  < 1/\delta$ for $p \in [0,M]$.
\end{enumerate}

\subsection{Theorem statements}
In the statement of the following theorem we mark quantities related to the
single customer strategy by $s$, and quantities related to the
bundling strategy by $b$. E.g., $U_i^s$ is the seller's utility from
customer $i$ using the single customer strategy, and $u_i^b$ is the
seller's expected utility from customer $i$ using the pair bundling
strategy.
\begin{theorem}
  \label{thm:bundling-better}
  Let $\{1,2\}$ be a pair of customers with valuation distributions
  satisfying the smoothness and boundedness conditions
  in~\ref{sec:smoothness}.  Let
  \begin{align*}
    u^s(p_1,p_2) = u_1^s(p_1)+u_2^s(p_2)
  \end{align*}
  be the seller's total expected utility when using the single
  customer strategy with prices $p_1$ and $p_2$.  Let
  \begin{align*}
    u^b(a_1,a_2,b)= u_1^b(a_1,a_2,b)+u_2^b(a_1,a_2,b) 
  \end{align*}
  be the seller's total expected utility when using the pair bundling
  strategy with prices $a_1$, $a_2$ and $b$. Then
  \begin{align}
    \label{eq:bundle-good}
    \max_{a_1,a_2,b} u^b(a_1,a_2,b) > \max_{p_1,p_2}u^s(p_1,p_2).
  \end{align}
\end{theorem}
That is, the best bundling strategy is strictly better than the best
single customer strategy.

The next theorem shows that when valuations are bounded then, as the
size of the bundle grows, the expected utility of the seller from the
customers approaches the sum of their expected valuations.
\begin{theorem}
  \label{thm:more-bundling-better}
  Consider a set of $n$ customers with valuation distributions
  satisfying the smoothness and boundedness conditions
  in~\ref{sec:smoothness}. 

  Let $\mu_i = \E{V_i}$ be customer $i$'s expected valuation, and let
  $\mu = \sum_{i=1}^n\mu_i$ be the sum of the customers' expected
  valuations.  Let
  \begin{align*}
    u^n(\bar{a},b)= \sum_{i=1}^nu_i^b(\bar{a},b) 
  \end{align*}
  be the seller's total expected utility when bundling all $n$
  customers  with prices $\bar{a} = (a_1,\ldots,a_n)$ and $b$.

  Then the seller's total expected utility satisfies
  \begin{align}
    \label{eq:big-bundle-good}
    \max_{\bar{a},b} u^n(\bar{a},b)
    \geq\left(1-\frac{4}{\delta}\sqrt{\frac{\log
          n}{n}}-O\left(\frac{1}{n}\right)\right)\mu
  \end{align}
\end{theorem}
Note that since a customer will never pay more than its valuation then
\begin{align*}
  \max_{\bar{a},b} u^n(\bar{a},b) \leq \mu.
\end{align*}

\section{Proofs}
\begin{proof}[Proof of Theorem~\ref{thm:bundling-better}]
  Let $(p_1,p_2)$ be the prices that maximize the seller's total
  expected utility for the single customer strategy. We will prove the
  theorem by showing that there exists $\eps > 0$ such that
  \begin{align}
    \label{eq:eps-helps}
    u^b(p_1+\eps,p_1,p_1+p_2) > u^s(p_1,p_2).
  \end{align}

  Note that since $p_1$ is optimal for the single customer strategy, then
  \begin{align*}
    \frac{\partial u^s(p_1,p_1)}{\partial
      p_1}=\frac{du_1^s(p_1)}{dp_1}=0.
  \end{align*}
  Hence there exist a constant $C_1$ such that for all $\eps$ small
  enough it holds that
  \begin{align}
    \label{eq:u-s-maximum}
    u^s(p_1+\eps,p_2) > u^s(p_1,p_2) - C_1\eps^2.
  \end{align}

  Let $B$ denote the event that the customers buy the bundle. Recall
  that in the bundling strategy with prices $(p_1+\eps,p_2,p_1+p_2)$
  $B$ occurs if and only if there exist $P_1$ and $P_2$ such that
  \begin{align}
    \label{eq:bundle-buy-sum}
    P_1+P_2 &= p_1+p_2,\\
    V_1 &\geq P_1, \nonumber \\
    V_2 &\geq P_2, \nonumber \\
    P_1 &\leq p_1+\eps, \nonumber \\
    P_2 &\leq p_2. \nonumber    
  \end{align}
  Using Eq.~\eqref{eq:bundle-buy-sum} we can substitute
  $P_2=p_1+p_2-P_1$ and arrive at the following equivalent condition:
  $B$ occurs if and only if there exists a $P_1$  such that 
  \begin{align*}
    p_2-V_2 &\leq P_1-p_1 \leq V_1-p_1,\\
    0 &\leq P_1-p_1 \leq \eps.
  \end{align*}
  In this form it is apparent that $B$ occurs if and only if
  \begin{align*}
    V_1+V_2 &\geq p_1+p_2,\\
    V_1 &\geq p_1,\\
    V_2 &\geq p_2-\eps.
  \end{align*}

  We now partition our probability space into the disjoint events
  $\{A_1^\eps,A_2^\eps,A_3^\eps,A_4^\eps,A_5^\eps\}$ (see
  Fig.~\ref{fig:events}), where the $\eps$ in the superscripts denotes
  the fact that these events depend on $\eps$. We compare
  $U^s=U^s(p_1,p_2)$ and $U^b=U^b(p_1+\eps,p_2,p_1+p_2)$ in each
  event.

  \begin{figure}[htp]
    \centering
    \includegraphics{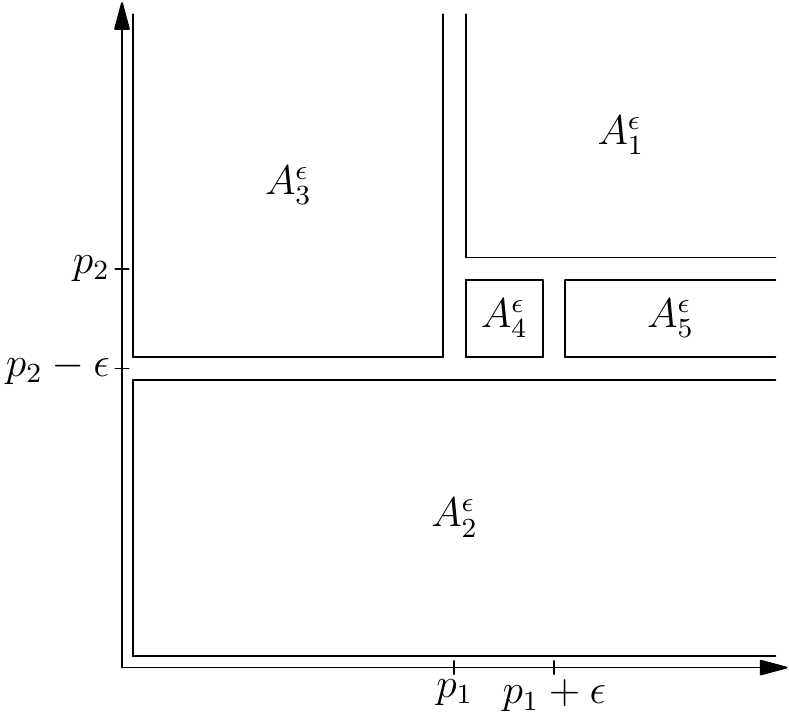} 
    \caption{Disjoint union of ${\R^+}^2$ into $\{A_i^\eps\}_{i=1}^5$.}
    \label{fig:events}
   \end{figure}

  \begin{enumerate}
  \item Let $A_1^\eps$ be the event that $(V_1,V_2) \in [p_1,\infty) \times
    [p_2,\infty)$. Then in $A_1$, in the single customer strategy both
    customers buy an item for a total of $p_1+p_2$, and in the
    bundling strategy the customers buy the bundle for
    $p_1+p_2$. Hence
    \begin{align*}
      U^s\ind{A_1^\eps} = U^b \ind{A_1^\eps}
    \end{align*}
    and
    \begin{align}
      \label{eq:region-A-1}
      \E{U^b\ind{A_1^\eps}} = \E{U^s\ind{A_1^\eps}}.
    \end{align}

  \item Let $A_2^\eps$ be the event that $(V_1,V_2) \in [0,\infty)
    \times [0,p_2-\eps)$. In this region the bundle is not bought, and
    neither does customer 2 buy an item on their own, in either
    strategies. Hence in this region customer 1 buys the item iff $V_1
    \geq p_1+\eps$ in the bundling strategy. Since $V_1$ and $V_2$ are
    independent then the (expected) utility for the seller in the
    bundling strategy is identical to what it would be when offering a
    single item to customer 1 for $p_1+\eps$. Since, by
    Eq.~\eqref{eq:u-s-maximum}, this expected utility is maximized
    when the price is $p_1$ (as is done in the single item strategy),
    then
    \begin{align*}
      \CondE{U^b}{A_2^\eps} \geq \CondE{U^s}{A_2^\eps} - C_1\eps^2.
    \end{align*}
    and
    \begin{align}
      \label{eq:region-A-2}
      \E{U^b\ind{A_2^\eps}} \geq \E{U^s\ind{A_2^\eps}} -
      C_1\eps^2\P{A_2^\eps}.
    \end{align}

  \item Let $A_3^\eps$ be the event that $(V_1,V_2) \in [0,p_1) \times
    [p_2-\eps,\infty)$. In this region the bundle is not bought, and
    neither does customer 1 buy an item on their own, in either
    strategies. Customer 2, however, buys in both
    strategies. Therefore the seller's utility is identical in this
    region:
    \begin{align*}
      U^s\ind{A_3^\eps} = U^b \ind{A_3^\eps}
    \end{align*}
    and
    \begin{align}
      \label{eq:region-A-3}
      \E{U^b\ind{A_3^\eps}} = \E{U^s\ind{A_3^\eps}}.
    \end{align}

  \item Let $A_4^\eps$ be the event that $(V_1,V_2) \in
    [p_1,p_1+\eps)\times[p_2-\eps,p_2)$. In this case we note that
    \begin{align*}
      \E{U^b\ind{A_4^\eps}} = \CondE{U^b}{A_4^\eps}\P{A_4^\eps},
    \end{align*}
    and
    \begin{align*}
      \E{U^s\ind{A_4^\eps}} = \CondE{U^s}{A_4^\eps}\P{A_4^\eps}.
    \end{align*}
    Now, since we assumed that the distribution of $(V_1,V_2)$ is
    non-atomic and since both $U^b$ and $U^s$ are bounded then there
    exists a constant $C$ such that for $\eps$ small enough it holds
    that $\P{A_4^\eps} < C\eps^2$, and so there exists a constant
    $C_2$ such that
    \begin{align}
      \label{eq:region-A-4}
      \E{U^b\ind{A_4^\eps}} \geq \E{U^s\ind{A_4^\eps}} -
      C_2\eps^2
    \end{align}
    for $\eps$ small enough.

  \item Finally, let $A_5^\eps$ be the event that $(V_1,V_2) \in
    [p_1+\eps,\infty)\times[p_2-\eps,p_2)$. Here in the single customer
    strategy customer 1 buys an item for $p_1$ and customer 2 does not
    buy. In the bundling strategy the customers purchase a bundle for
    $p_1+p_2$. Hence
    \begin{align*}
      \E{U^s\ind{A_5^\eps}} = p_1\P{A_5^\eps}
    \end{align*}
    and
    \begin{align*}
      \E{U^b\ind{A_5^\eps}} = (p_1+p_2)\P{A_5^\eps}.
    \end{align*}
    Since the distribution of $(V_1,V_2)$ is supported on $[0,M]^2$,
    and since $p_2 > 0$ (see note at the end of
    Section~\ref{sec:single-model}), then there exists a constant
    $C_3$ such that for $\eps$ small enough $\P{A_5^\eps} >
    C_3\eps$. Hence
    \begin{align}
      \label{eq:region-A-5}
      \E{U^b\ind{A_5^\eps}} \geq \E{U^s\ind{A_5^\eps}} +p_2C_3\eps.
    \end{align}
    for $\eps$ small enough.  
  \end{enumerate}

  Since the events $\{A_i^\eps\}$ are disjoint and since
  $\P{\cup_i A_i}=1$ then
  \begin{align*}
    u^b=\E{U^b}=\sum_{i=1}^5\E{U^b\ind{A_i^\eps}},
  \end{align*}
  with a similar expression for $u^s$. Therefore, as a conclusion of
  Eqs.~\eqref{eq:region-A-1}, \eqref{eq:region-A-2},
  \eqref{eq:region-A-3}, \eqref{eq:region-A-4} and
  \eqref{eq:region-A-5} we have that for $\eps$ small enough
  \begin{align*}
    u^b(p_1+\eps,p_2,p_1+p_2) \geq
    u^s(p_1,p_2)-(C_1\P{A_2^\eps}+C_2)\eps^2+p_2C_3\eps,
  \end{align*}
  and therefore for $\eps$ small enough
  \begin{align*}
    u^b(p_1+\eps,p_2,p_1+p_2) > u^s(p_1,p_2).
  \end{align*}
  
\end{proof}

\begin{proof}[Proof of Theorem~\ref{thm:more-bundling-better}]
  Consider the bundling strategy with individual prices $a_i=\infty$
  for all $i \in [n]$ (i.e., no single item sales) and $b =
  \mu- 2M\sqrt{n\log n}$. In this case the customers will either buy
  the bundle if its cost is less than the sum of their valuations, and
  buy nothing at all otherwise. Denote the sum of their valuations by
  $V=\sum_{i=1}^nV_i$.

  Since $V_i \in [0,M]$ then $\E{V_i^2}\leq M^2$. Hence by a version
  of Bernstein's
  inequality~\cite{bernstein1924modification}\footnote{We use the
    following version of Bernstein's inequality: Let $X_1,\ldots,X_n$
    be independent random variables such that $\E{X_i}=0$ and
    $|X_i|<M$ for all $i$. Then for any $t>0$ it holds that
    \begin{align*}
      \P{\sum_iX_i >t} \leq \exp\left(-\frac{t^2/2}{\sum_i\E{X_i^2}+Mt/3}\right).
    \end{align*}
 }   
    we have that,
  \begin{align*}
    \P{V < b} &\leq
    \exp\left(-\frac{4M^2n\log n}{2M^2n+2M^2\sqrt{n\log n}/3}\right)
  \end{align*}
  and hence
  \begin{align*}
    \P{V < b} &\leq \frac{1}{n}.
  \end{align*}
  
  Since the customers buy the bundle when $V \geq b$ then the seller's
  expected utility equals $\P{V \geq b}b$ and it holds that 
  \begin{align*}
    \P{V \geq b}b &\geq
    \left(1-\frac{1}{n}\right)\left(\mu-2M\sqrt{n\log n}\right).
  \end{align*}
  Since $f_i > \delta$ in the interval $[0,M]$ then
  $\E{V_i}>M\delta/2$ and $\mu > nM\delta/2$, and so it holds that
  \begin{align*}
    \P{V \geq b}b
    &\geq \left(1-\frac{1}{n}\right)\left(1-\frac{4}{\delta}\sqrt{\frac{\log n}{n}}\right)\mu\\
    &\geq \left(1-\frac{4}{\delta}\sqrt{\frac{\log n}{n}}-O\left(\frac{1}{n}\right)\right)\mu,
  \end{align*}
  where the second inequality follows from the fact that $\mu \geq
  n\eps$.

  Since the optimal strategy yields at least as much utility to the
  seller as this one, then 
  \begin{align*}
    \max_{\bar{a},b}u(\bar{a},b) \geq \left(1-\frac{4}{\delta}\sqrt{\frac{\log n}{n}}-O\left(\frac{1}{n}\right)\right)\mu.
  \end{align*}
\end{proof}

\section{Conclusion}
We showed how sellers may maximize profits by offering bundles of
items to rational groups of customers.  Our work suggest a number
of future research directions we wish to mention.

\subsection{Optimal auctions and optimal profit}
Our results show that it suffices to bundle pairs of customers to
increase profits under mild conditions, and that if the customers are
bundled in large groups it is possible to extract profit which
approaches the theoretical bound, as the group size increase.  For
example - assume that there are $N$ individuals - $N/2$ are paired
into rational groups, another $N/3$ are partitioned into rational
groups of size $3$, and the rest $N/6$ are partitioned into rational
groups of size $6$.  Assuming that all valuations are drawn
i.i.d.\ from the same distribution $F$ - what is the optimal auction
and by how much is it better than the single item auction?

\subsection{Overlapping Rational Groups and Social Networks}
The problem presented in the previous subsection can be generalized
further to a situation where an individual may belong to more than one
rational group: for example an individual may belong to a family and
to a small start-up company. The two groups are rational and are
offered different bundles. Understanding the optimal auction in this
setup and its relationship to the social network structure is, in our
opinion, an interesting open problem.

\bibliographystyle{abbrv} \bibliography{bundling}

\end{document}